\documentclass[11pt,reqno]{amsproc}

\usepackage[margin=1.0in]{geometry}
\usepackage[english]{babel}
\usepackage[utf8]{inputenc}
\usepackage{ulem}
\usepackage{amsmath,amssymb,amsthm, comment, mathtools}
\usepackage{hyperref,mathrsfs,xcolor}
\usepackage{enumitem}

\usepackage{genealogytree}

\newcommand{\R}{\mathbb{R}}

\newcommand{\ve}{\varepsilon}
\newcommand{\rmd}{{\rm d}}

\usepackage[mathscr]{euscript}

\newcommand{\eee}{equation}
\newcommand{\be}{\begin{\eee}}
\newcommand{\ee}{\end{\eee}}

\numberwithin{equation}{section}
\newtheorem{lemma}{Lemma}[section]
\newtheorem{prop}[lemma]{Proposition}
\newtheorem{theorem}[lemma]{Theorem}
\newtheorem{cor}[lemma]{Corollary}

\theoremstyle{definition}
\newtheorem{remark}[lemma]{Remark}

\newtheorem{example}[lemma]{Example}

\mathtoolsset{showonlyrefs,showmanualtags}

\title{Singular vortex pairs follow magnetic geodesics}

\author{Theodore D. Drivas}
\address{Department of Mathematics, Stony Brook University, Stony Brook, NY, 11790}
\email{tdrivas@math.stonybrook.edu}

\author{Daniil  Glukhovskiy}
\address{Department of Mathematics, Stony Brook University, Stony Brook, NY, 11790}
\email{daniil.glukhovskiy@stonybrook.edu }

\author{Boris Khesin}
\address{Department of Mathematics, University of Toronto, ON M5S 2E4, Canada}
\email{khesin@math.toronto.edu}

\begin{document}
\maketitle

\vspace*{-10mm}
\begin{abstract}
We consider pairs of point vortices having circulations $\Gamma_1$ and $\Gamma_2$ and confined to a two-dimensional surface $S$.  In the limit of zero initial separation $\ve$, we prove that they follow a magnetic geodesic in unison, if properly renormalized.  Specifically, the ``singular vortex pair" moves as a single charged particle on the surface with a charge of order $1/\ve^2$ in a magnetic field $B$ which is everywhere normal to the surface and of strength  $|B|=\Gamma_1 +\Gamma_2$.   In the case $\Gamma_1=-\Gamma_2$, this gives another proof of Kimura's conjecture \cite{kimura} that singular dipoles follow geodesics.   
\end{abstract}
\vspace*{2mm}
\maketitle

\section{Introduction}

One of the most classical areas of hydrodynamics is the study of the motion of point vortices in 2D ideal fluids. 
In this paper we are describing the limiting motion of a vortex pair on arbitrary surfaces. Let $S$ be a closed surface
embedded in $\R^3$ with an induced area form $\mu_S$. (For most of our considerations one does not need the embedding and can trace the vortex motion on an abstract 2D surface, but we assume it now to simplify the introduction.)
Consider a pair of vortices located at $z_1, z_2\in S$ and the singular vorticity form
\be\label{2vortex}
	\omega (x)=  \Gamma_1\delta_{z_1(t)}(x) + \Gamma_2\delta_{ z_2(t)}(x) - \frac{\Gamma_1+\Gamma_2}{{\rm Area}(S)}\,.
\ee
Here the constant term is understood as a multiple of the area form, $(\Gamma_1+\Gamma_2)\mu_S/{\rm Area}(S)$, which
ensures that the vorticity form is exact, i.e. it has zero total integral over $S$. (By abusing the notation we omit $\mu_S$ below.) Vortex positions evolve in time according to the Kirchoff--Helmholtz equations, see \S \ref{vortexsystem}. We are interested in the dynamical properties of the vortices as the width of the vortex pair tends to zero:
\be
\|z_1(0)- z_2(0)\| =\ve \to 0\,,
\ee
where $\|\cdot \|$ is the distance in $\R^3$.
In the case of a dipole $\Gamma_1=-\Gamma_2$, Kimura conjectured (and proved for surfaces of constant curvature) that \textit{when the width tends to zero the (singular) dipole moves along a geodesic on the surface} \cite{kimura}.  As such, Kimura noted that singular vortex dipoles are ``curvature checkers". The general case  was proved by Boatto and Koiller \cite{boatto2015} using Gaussian geodesic coordinates (see also  \cite{grotta2023interplay}) and  by Gustafsson \cite{gustafsson2022vortex} using complex analytic techniques.  

In this paper, we extend this picture to more general \textit{vortex pairs with different circulations}.  
We prove that singular vortex pairs evolve as a charged  particle confined to $S$ solely 
under the influence of a magnetic field of constant strength proportional to the sum $\Gamma_1+\Gamma_2$ and 
everywhere normal to the surface.  As such, general  vortex pairs  follow magnetic geodesics. Only in the case of the exact dipole  $\Gamma_1=-\Gamma_2$ does the magnetic effect disappear and it follows the ordinary geodesic.  

Recall that a curve $X: \mathbb{R}\to S$ parametrized by $t\in \R$ is a {\it magnetic geodesic} if it is a solution of the following
\textit{magnetic geodesic equation}:
      \begin{align}\label{geodesiceqn}
 \ddot{X }(t)&=    \mathrm{I\!I}_{X(t)} ( \dot{X}(t),\dot{X}(t)) +  \mathsf{q} \, \dot{X}(t) \times    B( X(t))
  \end{align}
  where $ \mathrm{I\!I}_{x} ( v,u)$ is the second fundamental form of the surface $S\subset \R^3$, a constant $\mathsf{q}$
  is the charge, and $B$  is the magnetic vector field, which is everywhere orthogonal to the surface.
  The geodesic is defined by setting the initial conditions $X(0)$ and $\dot{X}(0) $.

Our main result is the following statement:

\begin{theorem} \label{mainthm}
Let $z_1^\ve(t)$ and $z_2^\ve(t)$ be two point vortices on $S$ such that $\|z_1^\ve(0)- z_2^\ve(0)\| =\ve$. 
 The paths $\{z_i^\ve(t) \}_{t\in \mathbb{R}} $ have following asymptotics: for any fixed moment $t\in \mathbb{R}$ they tend to the corresponding magnetic geodesic as $\ve\to 0$, provided that renormalizations of $\Gamma_i$ enforce 
   \be\label{scalingrequirements}
 \frac{\Gamma_2- \Gamma_1}{\ve}= O(1) \qquad \text{and} \qquad \frac{\Gamma_1 +\Gamma_2}{\ve^2}    = O(1) \qquad \text{as} \qquad \ve\to 0.
  \ee
 Namely, 
there is a family of curves $X_\ve:\mathbb{R}\to S$ such that for $i=1,2$
$$
\| {z }_{ i}^\ve(t)-X_\ve(t)\|\to 0\quad \text{as} \quad \ve\to 0
$$ 
where  $X_\ve(t)$ solves the magnetic geodesic equation \eqref{geodesiceqn} for the  magnetic field 
\be
B(x) =({\Gamma_1 +\Gamma_2}) \hat{n}_S( x),
\ee
normal to the surface $S$
and the charge   $\mathsf{q}=\frac{1}{2\pi \ve^2}$. The initial conditions at $t=0$
 are 
\be
X_\ve(0)= \mathsf{p}, \qquad \dot{X}_\ve(0) =  \left(\frac{\Gamma_2- \Gamma_1}{2\ve}\right) \mathsf{w}\times \hat{n}_S(\mathsf{p}),
\ee
where  $\mathsf{p} = \lim_{\ve \to 0} \frac{1}{2}(z_1^\ve(0)+z_2^\ve(0))$ is the midpoint of the vortex pair and $\mathsf{w}= \lim_{\ve \to 0} \frac{z_1^\ve(0)-z_2^\ve(0)}{\|z_1^\ve(0)-z_2^\ve(0)\|}$ is the initial unit separation vector.  
 \end{theorem}

	\begin{figure}[h!]
		\centering
			\includegraphics[width=0.65\textwidth]{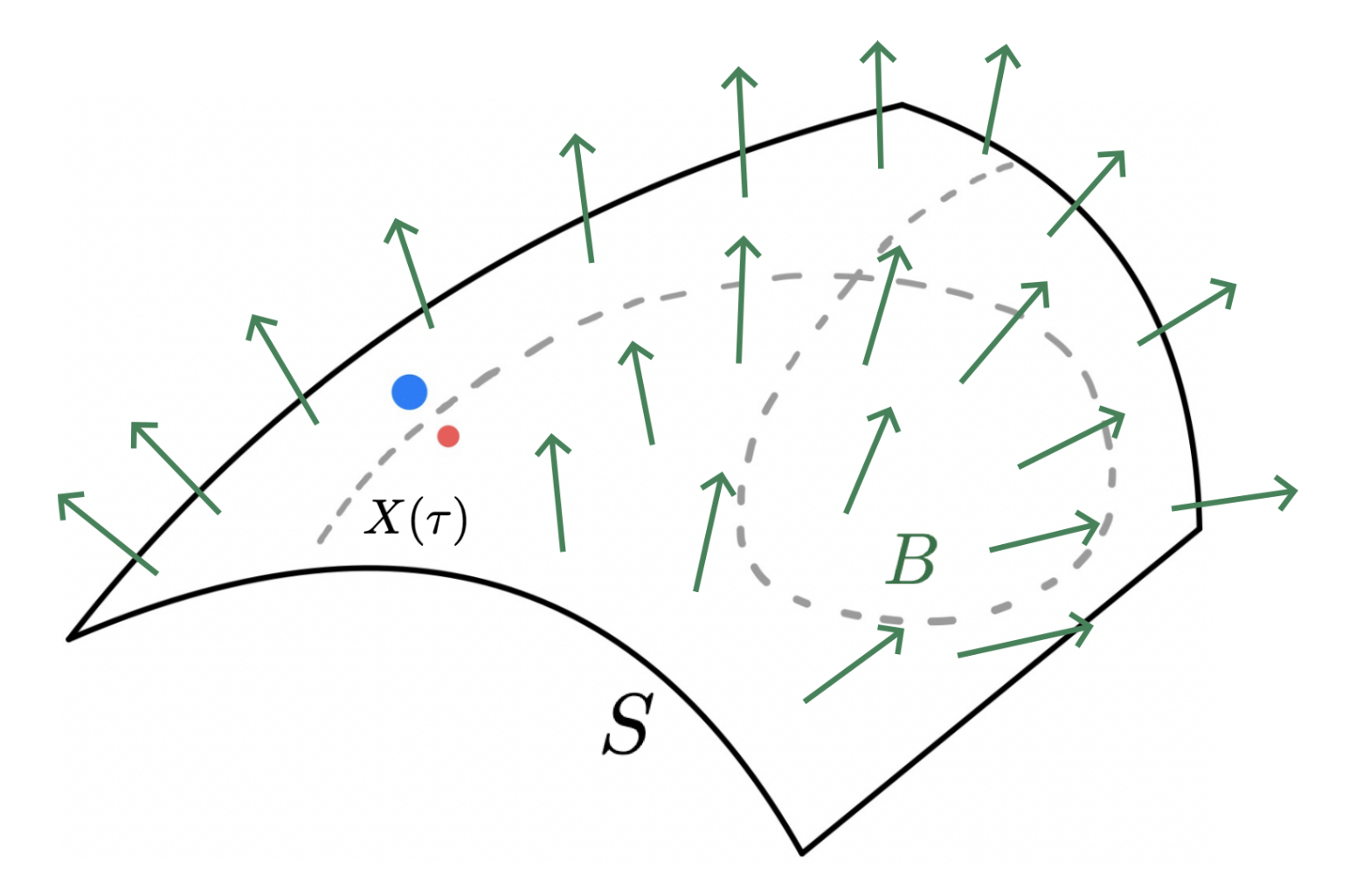} 
					\caption{An asymmetric vortex pair follows a magnetic geodesic on a surface.} 
	\end{figure}

  \begin{cor}
For a dipole, $\Gamma_1=- \Gamma_2$, the magnetic field vanishes and Eqn. \eqref{geodesiceqn} describes the geodesics on $S$. Thus, in the limit $\ve\to 0$, dipoles move along the geodesics on the surface $S$.
  \end{cor}
 
Theorem \ref{mainthm} can be understood as follows. Two vortices at a distance of order $\ve$  move each other with the velocity of order $\Gamma/\ve$. So the initial speed of the limiting magnetic geodesic is also of order $\Gamma/\ve$. 
 At the time of order $O(1)$ they could diverge at the distance of order $\Gamma/\ve$. By aligning their initial velocities this 
 distance is made of order $O(1)$. The theorem says that if, in addition, the acceleration satisfies the magnetic geodesic equation, then
 this distance will be of order $\ve$, i.e. given by the next term in the $\ve$-expansion. In this context, this is equivalent to the statement that  the latter distance goes to zero as $\ve\to 0$.
\newpage
 
 It is interesting to compare this setting with the gyroscope motion described by Cox and Levi \cite{Cox2016GaussianCA}. 
 The averaged motion of a gyroscope follows a magnetic geodesic in a magnetic field proportional to the Gaussian curvature of the surface $S$. Here, for the vortex pair, the magnetic field is constant and proportional to the sum of vortex strengths, $\Gamma_1+\Gamma_2$.

   \begin{remark}[Units]\label{units}  Note that circulation has units of (vorticity)$\times$(area) = $\mathsf{L}^2/\mathsf{T}$.   The units of the magnetic field (setting mass $\mathsf{m}=1$) are  $1/(\mathsf{T}^2 \mathsf{Amp})=1/(\mathsf{T}\mathsf{C})$ where $\mathsf{C}$ denotes one unit of charge $\mathsf{q}$. In our ``magnetic dictionary"
   the coulomb charge units translate to the inverse area: $\mathsf{C}\equiv1/\mathsf{L}^2$. The quantity $\mathsf{q}B$ has units $1/\mathsf{T}$.  
       \end{remark}
       
   \begin{remark}[Scalings]\label{scaling} 
  In order for the initial speed and  $\mathsf{q}B$ to be finite as $\ve \to 0$, one requires \eqref{scalingrequirements}.
This could be accomplished, for instance, by taking  
\be\label{Gammascaling}
\Gamma_i = c_i \mathsf{v_0} \ve
\ee
where $\mathsf{v_0}$ is fixed with units of velocity and $c_1$ and $c_2$ are dimensionless functions of $\ve$ satisfying 
  \be\label{cscale}
  c_1=c_1(\ve) =  \frac{\mathsf{c}}{2} \ve+1, \qquad   c_2=c_2(\ve) =  \frac{\mathsf{c}}{2} \ve-1.
  \ee
  One should imagine here a limit where circulations are scaled down in proportion to the inter-vortex distance, while simultaneous becoming closer to a vortex dipole at a rate which maintains a finite effect of the Lorentz force.  
       \end{remark}

\begin{example}[Motion on the plane]
{\rm
On the plane, two point vortices with circulations $\Gamma_1$ and $\Gamma_2$ sitting at $z_1$ and $z_2$ rotate about the center of vorticity located at
$z_c= (z_1 \Gamma_1 + z_2 \Gamma_2)/(\Gamma_1 + \Gamma_2)$ with angular velocity 
$\xi = (\Gamma_1 + \Gamma_2)/(2\pi \ve^2)$, where $\ve$ is the pair width $\ve=\|z_1 - z_2\|$, see e.g. \cite{newton2001n}.
In the magnetic field $B$ the angular velocity of a particle of mass $\mathsf{m}$ and charge $\mathsf{q}$ is 
$\xi = \mathsf{q}B/\mathsf{m}$, i.e. $B$ is proportional to $\xi$. For $\Gamma_1=-\Gamma_2$, we obtain $ B=0$ and hence we get a standard geodesic (a straight line) for a pure vortex dipole.   To keep the circular orbit having finite non-zero radius as $\ve \to 0$, we may scale $\Gamma_i$ as in \eqref{Gammascaling}, \eqref{cscale}.  In this case, the limit is just circular motion around point
$
z_c^* = \frac{1}{\mathsf{v_0}} \mathsf{w} + \mathsf{p}
$
with angular velocity 
$
\xi =  \frac{\mathsf{v_0}}{2\pi},
$
  where $\mathsf{w}$ is the (normalized) initial separation vector of the vortex pair.
This corresponds to a charged particle in a magnetic field of strength proportional to $\xi$, initial position $\mathsf{p}$ and initial velocity $\mathsf{w}^\perp$.
}
\end{example}

  \vspace{4mm}


\section{Point vortex motion on surfaces}
\label{vortexsystem}

Let $S$ be a closed surface of genus $\mathfrak{g}=0$ embedded in $\mathbb{R}^3$. See Remark \ref{homologicalpoint} for more general surfaces, which require consideration of evolving harmonic velocities. We recall the main properties of  the motion of point vortices on such surfaces,  see e.g. \cite{boatto2015, kimura1987vortex,sakajo2016point,dritschel2015motion,grotta2023interplay}, as well as \cite{newton2001n} for a comprehensive review. The evolution equations for such point vortices were first derived by Helmholtz as a finite-dimensional approximation of a two-dimensional ideal fluid. They describe the motion of vorticity that is a finite linear combination of Dirac $\delta$-functions located at points  $z_i \in S$ and  carrying  circulations $\Gamma_i\in \mathbb{R}$ in a constant background vorticity field.  Specifically, 
\be\label{vortdist}
\omega(x,t) = \sum_{i=1}^N \Gamma_i \Big(\delta_{z_i(t)}(x)  - \frac{1}{{\rm Area}(S)}\Big).
\ee
The  constant  vorticity  $ \frac{1}{{\rm Area}(S)}  \sum_{i=1}^N \Gamma_i$  ensures the mean-zero condition:
$
\iint_S \omega(t)  \mu_S =0,
$
as required for $\omega$ to represent the curl of a vector field.  

Let $\tilde{\Delta}$ denote the Laplace-Beltrami operator on $S$ with respect to the metric induced by embedding of $S$ into $\mathbb{R}^3$. The Green function  $G_S$ for $\tilde{\Delta}$ is characterized by 
\newpage
\begin{align}\label{greensfunction}
\tilde{\Delta}_p G_S(p,q) &= \delta_q(p) - \frac{1}{{\rm Area}(S)},\\
 G_S(p,q)&= G_S(q,p),\\
 \iint_S G_S(p,q) &\  \mu_S(q)  =0 ,\\
  G_S(p,q)\ + \ &\frac{1}{2\pi}\log d_S(p,q)\qquad \text{is bounded}
\end{align}
where $d_S(p,q)$ is the geodesic distance with respect to the induced metric and   $\mu_S(q)$ is the induced area form on $S$, see \cite{boatto2015, okikiolu2008negative,flucher1997vortex}.  Such $G_S$ is the kernel for the integral operator solving Poisson's equation 
\be
\Delta^{-1} f(p) = \iint_S G_S(p,q) f(q)  \mu_S(q) .
\ee
To write the equations of motion for the vortices, we must introduce a regular part of $G$ (sometimes called the \textit{Robin function} in the literature)
\be
R_S(p, q) := 
\begin{cases}
	G_S(p, q) + \frac{\log d_S(p,q)}{2 \pi} & p\neq q\\
	\lim_{p \to q} \big[ G_S(p, q) + \frac{\log d_S(p,q) }{2 \pi}\big] &  p = q
\end{cases}.\label{robinfunction}
\ee
Vortex dynamics for $N$ vortices \eqref{vortdist} are then Hamiltonian with the Hamiltonian function 
\be\label{pairHamilton}
H(z_1, \dots, z_N) = \sum_{ 1\leq i< j \leq N} \Gamma_i \Gamma_j G_S(z_i, z_j) + \frac{1}{2}\sum_{\ell=1}^N  \Gamma_\ell^2 R_S(z_\ell,z_\ell),
\ee
and the symplectic form given by the weighed combination of the induced area form $\mu_S$ on $S$:
\be
\Omega(z_1, \dots, z_N) =  \sum_{\ell=1}^N \Gamma_\ell  \mu_S(z_\ell).
\ee
For a pair of vortices, 
the Hamiltonian is simply 
\be
H(z_1,z_2)=\Gamma_1 \Gamma_2 G_S(z_1,z_2) + \frac{1}{2}\Gamma_1^2 R_S(z_1,z_1) + \frac{1}{2}\Gamma_2^2 R_S(z_2,z_2).
\ee 
 If $J$ represents the almost complex structure on $S$ (rotation by ninety degrees in the tangent plane), the equations of motion for $z_1$ and $z_2$ read
$$
\Gamma_i \dot{z }_i(t) = J(z_i) \tilde{\nabla}_iH (z_1(t), z_2(t))\quad\text{for}\,\, i=1,2
$$
where $\tilde{\nabla}$ is the covariant derivative associated to  the induced metric on the surface.  

\begin{remark}[Surfaces of genus $\mathfrak{g}\geq 1$]\label{homologicalpoint}
The Hamiltonian system \eqref{pairHamilton} is the correct description only if the surface has trivial homology (genus $\mathfrak{g}=0$). If $\mathfrak{g}\geq 1$, then the harmonic component ($2\mathfrak{g}$ degrees of freedom) must be simultaneously evolved, and contribute to the evolution of the vorticity, see \cite{fluidcohomology} and \cite[\S 2.2]{drivas2023singularity}.  Recent work  \cite{grotta2023interplay} derives the correct Hamiltonian  point vortex dynamics which includes the fluid cohomology. However, since the velocity of the harmonic part is of lower order as the width tends to zero, it does not effect the result that the singular pair follows a magnetic geodesic. Thus, for simplicity or presentation, we omit this coupling.
\end{remark}

As described in the introduction, we will prove that in an appropriate coincidence limit of two vortices, the motion is that of magnetic geodesics on $S$.  For clarity and completeness, we now briefly review standard geodesic motion on $S$ before proceeding to our main result.


\section{Review: Geodesic motion on a surface}
\label{geodesicmotion}

Let us describe the motion of a geodesic on a immersed codimension-1 submanifold $S$ of $\mathbb{R}^n$ 
\be\label{dipolesurf}
S  := \{ x\in \mathbb{R}^n \ : \ f(x) =0\}
\ee
defined by the non-degenerate zero level set of  a  function $f:\mathbb{R}^n\to \mathbb{R}$. The unit normal
\be
 \hat{n}_S(x) := \frac{\nabla f(x)}{\|\nabla f(x)\|}
\ee
is well defined everywhere in a tubular neighborhood of the surface $S$.
The geodesic motion is that of an ideal mechanical particle, i.e. there is no friction force acting on the particle moving along the submanifold.
This mechanical system for a particle of mass $\mathsf{m}$ is defined by Newton's equations
\begin{align}\label{heb}
\mathsf{m} \ddot{X}(t) &=  \lambda(t) \nabla f(X(t)),\\
({X},\dot{X})|_{t=0}&= (X_0,V_0)\in T_{X_0}S, \label{hee}
\end{align}
where $\lambda = \lambda(t)$ is determined in order to satisfy the  constraint equation
\be \label{constaints}
 f (X(t)) =  f (X(0)).
\ee
Equations \eqref{heb}--\eqref{hee} subject to the constraints \eqref{constaints}  describe the geodesic motion on $S$.

\begin{lemma} 
The Lagrange multipliers $(t)$ are given by the formula
\be\label{lambdaform}
\lambda(t) :=-\frac{({\rm Hess} \ f )|_{X(t)}( \dot{X}(t),\dot{X}(t))}{\|\nabla f(X(t))\|^2}.
\ee
\end{lemma}

Indeed, to obtain equations for the Lagrange multipliers $\lambda:= \lambda(t)$, we differentiate  \eqref{constaints} twice:
\be \label{constaints2}
\ddot{X}(t)\cdot \nabla f(X(t))=-({\rm Hess} \ f )|_{X(t)}  ( \dot{X}(t),\dot{X}(t)),
\ee
which, as a consequence of \eqref{heb}, leads to
$
\|\nabla f(X(t))\|^2 \lambda(t) =- ({\rm Hess} \ f )|_{X(t)}  ( \dot{X}(t),\dot{X}(t)).
$
Since  $\nabla f|_{S}\neq 0$, the formula \eqref{lambdaform} follows. \medskip

 The system \eqref{heb}, \eqref{hee} and \eqref{lambdaform} complete the description of geodesic motion.
 We now relate this picture with a more geometric description. Recall that the orthogonal projection onto the tangent space to $S$ at $x$ is
\be\label{projector}
\mathbf{P}_x = I - \hat{n}_S(x)\otimes\hat{n}_S(x).
\ee
The {\it second fundamental form} of the submanifold $S $ is given by 
\be
\mathrm{I\!I}_x(v,w)= \mathbf{P}_x^\perp (\nabla_v w) \quad \text{for any vectors} \quad v, w\in T_xS 
\ee
and $ \mathbf{P}_x^\perp$ is the orthogonal projection onto the fiber at $x\in S$ of the normal bundle.  

\begin{lemma}\label{secondfundformlem}
Explicitly, the second fundamental form has the following expression: for $x\in S $
\be\label{secondfundform}
\mathrm{I\!I}_x(u,v) = -\frac{({\rm Hess} \ f )|_{x}(u,v)}{\|\nabla f(x)\|^2} \nabla f(x) .
\ee
\end{lemma}

Indeed, present
$
 \mathbf{P}_x^\perp (\nabla_v w)  =  \beta \nabla f 
$
for an appropriate $\beta\in \mathbb{R}$ to be determined. Taking the inner product with $\nabla f$, we find the identity
$\nabla f\cdot (\nabla_v w)  = \|\nabla f\|^2 \beta$.
Since $ w\cdot \nabla f|_x=0$ as $w\in T_xS $, we find
$
\nabla f\cdot (\nabla_v w)  =    -{\rm Hess}f (w,v)
$
and thus $\beta =-\frac{{\rm Hess}f (w,v)}{\|\nabla f\|^2}$. This implies formula \eqref{secondfundform} for 
$\mathrm{I\!I}_x$. 

\begin{remark}
 Introducing the momentum variable $P(t):= \mathsf{m}\dot{X}(t)$, the geodesic equations \eqref{heb}--\eqref{hee} may be  intrinsically (and without constraints) written as a first order system
\begin{align} \label{ge1}
\dot{X}(t) &= \tfrac{1}{\mathsf{m}} P(t),\\
\dot{P}(t)&=\mathrm{I\!I}_{X(t)} ( \tfrac{1}{\mathsf{m}}P(t), \tfrac{1}{\mathsf{m}}P(t)).  \label{ge2}
\end{align}
This is nothing but the equations for geodesic motion, which could be seen by noting $\mathbf{P}_{X(t)}\ddot{X}(t)=0$.
\end{remark}


\section{A geometric proof sketch of the Theorem}

Here we present a geometric reasoning for the validity of the main theorem. 
\medskip

Recall that any geodesic problem in Riemannian geometry can be reformulated
in terms of symplectic geometry. Namely, geodesics on a manifold $S$ (of any dimension) are extremals of a quadratic
Lagrangian on $TS$ (coming from the metric on $S$). They can also be described
by the Hamiltonian flow on $T^*S$ for the quadratic Hamiltonian function $(p,p)_q/2$
obtained from the kinetic Lagrangian $(\dot q,\dot q)_q/2$ (via the Legendre transform,  where $(\cdot , \cdot)_q $ are the inner products in the corresponding tangent and cotangent spaces at the point $q\in S$). Magnetic geodesics are by definition projections to $S$ of the Hamiltonian trajectories with the same Hamiltonian $(p,p)_q/2$ but instead of the canonical symplectic structure $\Omega_{\rm can}:= dp\wedge dq=\sum_i dp_i\wedge dq_i$ on $T^*S$ 
one considers the sum $\Omega_{\rm mag}:=\sum_i dp_i\wedge dq_i +\sum_{i<j}B(q)dq_i\wedge dq_j$ related to
the magnetic field $B(q)$ on $S$.
\medskip

Recall that the vortex pair with strengths $\Gamma_1$ and $\Gamma_2$ on $S$ is a Hamiltonian system on $S\times S$ 
with the symplectic structure $\Omega_{\rm pair}= \Gamma_1 dz_1\wedge d\bar z_1+ \Gamma_2 dz_2\wedge d\bar z_2$ and the Hamiltonian $H(z_1,z_2)=\Gamma_1 \Gamma_2 G_S(z_1,z_2) + R(z_1,z_2)$
where the Green function $G_S(z_1,z_2)$ on the surface $S$ has singularity of order $\log \|z_1-z_2\| $, while $R$
is a combination of the  (smooth) Robin functions, see \eqref{pairHamilton}.
\medskip

First consider the case of a pure vortex dipole, $\Gamma_1=-\Gamma_2=:\Gamma$ in the limit $\|z_1-z_2\| \to 0$. 
This means that we need to consider a neighborhood of the diagonal $\Delta=\{z_1=z_2\}\subset 
S\times S$. Note that this diagonal is a Lagrangian submanifold, i.e. $\Omega_{\rm pair}|_\Delta=0$, and hence its neighbourhood $U(\Delta)$ is symplectomorphic to the cotangent bundle $T^*\Delta$ with the symplectic structure $({\rm const})\cdot dp\wedge dz$ for $z:=z_1=z_2$ and $p:=z_1-z_2$, see \cite{ArnoldGiv}. In this neighbourhood $U(\Delta)\subset T^*\Delta$ the singular part of the Hamiltonian 
$H$ is $\Gamma^2\log\|p\|$ and it depends only on $\|p\|=\|z_1-z_2\|$, the distance to $\Delta$. 
Then the principal part of the corresponding Hamiltonian field $J\nabla H$ in $U$ is directed in the same way as the 
Hamiltonian field $J\nabla \widetilde H$ for the Hamiltonian function $\widetilde H:=\|p\|^2$ of the geodesic flow. 
Note, however, that the vortex pair field increases as $1/\|p\|=1/\|z_1-z_2\|$ as $\ve:=\|z_1-z_2\| = \|p\| \to 0$, while  
the geodesic field decays as $\|p\|$ as $\|p\| \to 0$. This explains the infinite speed of the vortex dipole as its width goes to zero
and the necessity of its renormalization. The renormalization boils down to the division by $\ve^2$ and it 
makes the  corresponding Hamiltonian fields of the dipole and the geodesic coincide in the principle order, which completes the proof. (One can trace similar type arguments in the proof below, as well as in \cite{boatto2015, gustafsson2022vortex}.)
\medskip

Now turn to a general vortex pair with different (not necessarily opposite) $\Gamma_1$ and $\Gamma_2$. 
The corresponding Hamiltonian is almost the same
(it differs from $H$ by a constant factor), but the symplectic structure  $\Omega_{\rm pair}$ restricted  
to the diagonal $\Delta$ is non-zero anymore, but it is $\Omega_{\rm pair}|_\Delta=(\Gamma_1+\Gamma_2)dz\wedge d\bar z$. 
According to the Givental--Weinstein theorem, a symplectic structure in a neighborhood of a submanifold in a symplectic space is fully defined
by the restriction of the symplectic structure to that submanifold, see \cite{ArnoldGiv}. Hence one can regard
a neighborhood  $U(\Delta)\subset S\times S$ as the cotangent bundle $T^*\Delta$ equipped  
 with a magnetic symplectic structure $\Omega_{\rm can} + (\Gamma_1+\Gamma_2)dz\wedge d\bar z$. 
 This structure is, in fact, the magnetic symplectic structure $\Omega_{\rm mag}$ above with the constant magnetic
 field $B=\Gamma_1+\Gamma_2$ on $S$. Since the Hamiltonian is (almost) the same as for the case of the dipole,
 the same consideration of its principal part is applicable here. Thus after a renormalization,  
 the corresponding Hamiltonian trajectories 
 are tracing magnetic geodesics on $\Delta$ in the principal order, as claimed.
\medskip

Note that in the geometric setting  above one does not need the surface $S$ to be embedded into the space $\R^3$, and the statement on a vortex pair tracing magnetic geodesics holds for an abstract two-dimensional Riemannian manifold.

\newpage

\section{Analytical proof of the Theorem}

As in Section~\ref{geodesicmotion}, we consider a closed embedded surface $S : = \{ x\in \mathbb{R}^3 \ : \ f(x) = 0\}$ of genus $\mathfrak{g}=0$  in three-space  defined by the zero set of a smooth function $f:\mathbb{R}^3\to \mathbb{R}$.
The considerations below are local and apply in far greater generality. For surfaces of higher genus $\mathfrak{g}\geq 1$, the point vortex system must be modified to account for an evolving harmonic velocity, see Remark \ref{homologicalpoint}.

Recall the unit normal to this surface is $\hat{n}_S(x) := \frac{\nabla f(x)}{\|\nabla f(x)\|}$. Let $\hat{\tau}_1(x)$ and $\hat{\tau}_2(x)$ be an orthonormal basis of the tangent space at a given point $x\in S$ arranged so that $(\hat{\tau}_1, \hat{\tau}_2,\hat{n})$ is a right triple.  Let $J(x)$ be the operator which preserves $ \hat{n}$ and
 takes $\hat{\tau}_1$ to $\hat{\tau}_2$, such that  $J^2=-I$, that is 
 \be\label{Jeqn}
J(x)  v = v\times \hat{n}_S(x) \qquad \text{for any } \quad x\in S, \ v \in T_x S.
\ee


 From the discussion in Section~\ref{vortexsystem}, for a pair of vortices the Hamiltonian is 
\be
H(z_1,z_2)=\Gamma_1 \Gamma_2 G_S(z_1,z_2) + \frac{1}{2}\Gamma_1^2 R_S(z_1,z_1) + \frac{1}{2}\Gamma_2^2 R_S(z_2,z_2).
\ee 
The following Lemma isolates the leading singular behavior of the Green function
\begin{lemma}\label{greenlem}
The Green function of the surface Laplacian \eqref{greensfunction} may be expressed as
\be
G_S(p,q)  = -\frac{1}{2\pi} \log d_S(p,q) + R_S(p,q)
\ee
where, for all $\alpha\in(0,1)$, there is a constant $C=C(S,\alpha)$ such that the ``regular part" \eqref{robinfunction} satisfies
\be
\|  R_S(\cdot,q) \|_{C^{2,\alpha}(S)} \leq C \qquad \text{for all} \qquad  q\in S.
\ee
\end{lemma}
\begin{proof}
For dimensions $d=3,4,5$, this result appears as Theorem 3.5 (a) of \cite{schoen1994lectures}.  Our setting of $d=2$ follows in the same fashion.  Here we provide a direct, simple proof for completeness.  Note
\begin{align}\nonumber
\tilde{\Delta}_p R_S(p,q) &= \tilde{\Delta}_p G_S(p,q)  + \frac{1}{2\pi}  \tilde{\Delta}_p\log d_S(p,q) \\ \label{lapRcalc}
&=\delta_q(p)- \frac{1}{{\rm Area}(S)}   + \frac{1}{2\pi}  \tilde{\Delta}_p\log d_S(p,q).
\end{align}
We now claim that 
\be
 \tilde{\Delta}_p\log d_S(p,q) = -\delta_q(p)-\frac{1}{3}K(p)  + O(d_S),
\ee
where $K(p)$ is the surface Gaussian curvature at $p$, and $O(d_S)$ denotes a term bounded by $\|\cdot \| \leq C d_S(p,q)$. Indeed,
setting $r=d_S(p,q)$, for points $p,q$ such that $r>0$
\begin{align}\label{lapdcal}
 \tilde{\Delta}_p\log d_S(p,q) &= \frac{1}{r} \left( \tilde{\Delta}_p d_S(p,q) - \frac{1}{r}\right)=\frac{1}{r} \left( H_{p,q}(q) - \frac{1}{r}\right)
\end{align}
where $H_{p,q}(q) $ is the  mean curvature of the geodesic sphere (here: circle) of radius $r$, see \cite[Prop. A.4]{beltran2019discrete}. Then there is an expansion (in dimension 2) of the mean curvature (see \cite[Lemma 3.4]{fan2007large}):
\be
H_{p,q}(q) = \frac{1}{r} - \frac{1}{3} R_{ij}(p) \frac{x^i x^j}{r} + O(r^2)
\ee
where $R_{ij}(p)$ is the Ricci tensor evaluated at the point $p$. In dimension 2, the Ricci tensor is proportional to the metric through the Gauss curvature,  $R_{ij}=K g_{ij}$. Moreover, at $p$ we have $g_{ij}(p)=\delta_{ij}$ so that $g_{ij}(p)x^i x^j=|x|^2=r^2$ since $q$ is on the geodesic circle centered at $p$.  This yields
\be
H_{p,q}(q) = \frac{1}{r} - \frac{1}{3} K(p) r + O(r^2).
\ee
Inserting into \eqref{lapdcal} and \eqref{lapRcalc}, we obtain
\be
\tilde{\Delta}_p R_S(p,q) = -\frac{1}{6\pi}K(p) - \frac{1}{{\rm Area}(S)}   + O(d_S). \label{Rrhs}
\ee
The right hand side of \eqref{Rrhs} is $C^{0,1}$.  It follows from elliptic regularity that $R_S(\cdot,q) \in C^{2,\alpha}(S)$.
\end{proof}

With Lemma \ref{greenlem} in hand, we write the Hamiltonian as
\be\label{hamreg}
H(z_1,z_2) =-\frac{\Gamma_1 \Gamma_2}{2\pi}  \log \| z_1-z_2\|  +  \mathsf{Reg}_S(z_1,z_2)
\ee
where the regular part $ \mathsf{Reg}_S(p,q)$ (bounded in $C^{2,\alpha}(S\times S)$) is
\be\label{eq:regular}
 \mathsf{Reg}_S(p,q):=- \frac{\Gamma_1 \Gamma_2}{2\pi} \log \frac{d_S(p,q)}{\|p-q\|}+  \Gamma_1 \Gamma_2  R_S(p,q) + \frac{1}{2}\Gamma_1^2 R_S(p,p) + \frac{1}{2}\Gamma_2^2 R_S(q,q).
\ee
By \eqref{Jeqn}, the symplectic gradient of any function $\phi$ can is represented as
$
J(z)\tilde{\nabla}\phi(z)= \nabla \phi(z) \times \hat{n}_S(z),
$
where $\nabla$ is the gradient in the ambient $\mathbb{R}^3$.
As such,  the equations of motion read
\begin{align*}
 \dot{z }_1 (t)&=  -\frac{\Gamma_2}{2\pi } \frac{z_1-z_2}{  \| z_1-z_2\|^2}\times \hat{n}_S(z_1)  + \frac{1}{\Gamma_1} \nabla_1 \mathsf{Reg}_S(z_1,z_2)\times \hat{n}_S(z_1) ,\\
 \dot{z }_2(t)&= \phantom{-} \frac{\Gamma_1}{2\pi  }  \frac{z_1-z_2}{ \| z_1-z_2\|^2}\times \hat{n}_S(z_2)  +\frac{1}{\Gamma_2}  \nabla_2 \mathsf{Reg}_S(z_1,z_2)\times \hat{n}_S(z_2).
\end{align*}
We introduce relative coordinates
\begin{align}\label{zdefs1}
z _{\rm abs}:= \tfrac{1}{2}(z_1+  z_2), &\qquad z _{\rm rel} := \tfrac{1}{2}(z_1-  z_2), \qquad w := \|z _{\rm rel}\|.
\end{align}
With these notations, we may write
\begin{align}\label{z1eqn}
 \dot{z }_1(t) &=  -\frac{1}{\pi  w^2 } \Big[\Gamma_2 z _{\rm rel}\times \hat{n}_S(z _{\rm abs})   + F_1(z _{\rm abs},z _{\rm rel})\Big],\\
\dot{z }_2(t)&= \phantom{-}  \frac{1}{\pi w^2} \Big[ \Gamma_1 z _{\rm rel}\times \hat{n}_S(z _{\rm abs})+   F_2(z _{\rm abs},z _{\rm rel})\Big],\label{z2eqn}
 \end{align}
 where we have defined
 \begin{align} \label{f1}
 F_1(a,b)&:=   \Gamma_2 b\times\Big( \hat{n}_S(a+b) - \hat{n}_S(a) \Big)  -    \tfrac{\pi}{\Gamma_1}  \|b\|^2 (\nabla_p \mathsf{Reg}_S)(a+b, a-b)\times \hat{n}_S( a+b),\\  \label{f2}
  F_2(a,b)&:=    \Gamma_1b\times\Big( \hat{n}_S(a-b)-\hat{n}_S(a)\Big) +  \tfrac{\pi}{\Gamma_2}  \|b\|^2 (\nabla_q \mathsf{Reg}_S)(a+b, a-b)\times \hat{n}_S(a-b).
 \end{align}
Combining \eqref{zdefs1}, \eqref{z1eqn}--\eqref{z2eqn} and using $z _{1} = z _{\rm abs} + z _{\rm rel}$ and $z _{2} = z _{\rm abs} - z _{\rm rel}$ yields

\begin{lemma} The  evolution for $z _{\rm abs}$, $z _{\rm rel}$ and $w= \|z _{\rm rel}\|$ read
\begin{align}\label{zabs1}
 \dot{z }_{\rm abs}(t) &= \phantom{-} \frac{1}{2\pi  w^2 } \Big[(\Gamma_1-\Gamma_2)z _{\rm rel}\times \hat{n}_S(z _{\rm abs})   + \mathsf{E}_{{z }_{\rm abs} }(z _{\rm abs},z _{\rm rel})\Big],\\ \nonumber
 \dot{z }_{\rm rel} (t) &=  -\frac{1}{2\pi w^2} \Big[(\Gamma_1+\Gamma_2)z _{\rm rel}\times \hat{n}_S(z _{\rm abs}) + (\Gamma_2-\Gamma_1)z_{\rm rel}\times (z_{\rm rel}\cdot \nabla\hat{n}_S(z_{\rm abs}))\\ \label{zrel1} 
 & \qquad \qquad \qquad + \mathsf{E}_{{z }_{\rm rel} }(z _{\rm abs},z _{\rm rel})\Big],\\ \label{w1}
 \dot{w}(t) &=  -\frac{1}{2\pi w^2 } \frac{{z }_{\rm rel} }{w}\cdot   \mathsf{E}_{{z }_{\rm rel} }(z _{\rm abs},z _{\rm rel}),
 \end{align}
 where 
 \be\label{Edefs}
\mathsf{E}_{{z }_{\rm abs} }(a,b):= F_2(a,b)- F_1(a,b), \quad \mathsf{E}_{{z }_{\rm rel} }(a,b):=F_1(a,b)+F_2(a,b) - (\Gamma_2-\Gamma_1)b\times (b\cdot \nabla\hat{n}_S(a)).
\ee
 \end{lemma}

Using the results of Lemma \ref{bndlem} in Appendix \ref{appendix}, we estimate lower--order contributions in the evolutions \eqref{zabs1}--\eqref{w1}. We use the big--$O$ notation to mean 
$f = O(g)$ provided $\|f\| \leq C \|g\|$ for a $g$--independent constant $C>0$. We arrive at the following bounds
 \begin{lemma}\label{lem:hard}
   With $\Gamma=  \max\{|\Gamma_1|,|\Gamma_2|\} $, the  evolution for $z _{\rm abs}$, $z _{\rm rel}$ and $w$ satisfy
\begin{align}
 \dot{z }_{\rm abs}(t) &= \phantom{-} \frac{1}{2\pi  w^2 } \Big[(\Gamma_1-\Gamma_2) {z} _{\rm rel}\times \hat{n}_S(z _{\rm abs})   + O(\Gamma w^3)+ O((\Gamma_1+\Gamma_2) w^2)\Big],\\ \nonumber
 \dot{z }_{\rm rel} (t) &=  -\frac{1}{2\pi w^2} \Big[(\Gamma_1+\Gamma_2) {z} _{\rm rel}\times \hat{n}_S(z _{\rm abs}) + (\Gamma_2-\Gamma_1)z_{\rm rel}\times (z_{\rm rel}\cdot \nabla\hat{n}_S(z_{\rm abs})) \\
 &\qquad\qquad\qquad+ O(\Gamma w^3)+ O((\Gamma_1+\Gamma_2) w^2)\Big],\\
 \dot{w}(t) &=  -\frac{1}{2\pi w^2}\Big[O(\Gamma w^4)+ O((\Gamma_1+\Gamma_2) w^3)\Big].
 \end{align}
 \end{lemma}
 Here we introduce the new time variable 
 \be
\tau(t) = \int_0^t  \frac{\ve^2}{w^2(s)}\rmd s
\ee
along with its inverse function $t(\tau)$. We thus have ${z}(t(\tau)) =\tilde{z}(\tau) $, and 
abusing notation we write ${z}(\tau) :=\tilde{z}(\tau) $.  In this new time, we have
  \begin{align}\label{zabseqn}
 \dot{z }_{\rm abs}(\tau) &= \phantom{-} \frac{1}{2\pi \ve^2 } \Big[(\Gamma_1-\Gamma_2)z _{\rm rel}\times \hat{n}_S(z _{\rm abs})   + \mathsf{E}_{{z }_{\rm abs} }(z _{\rm abs},z _{\rm rel})\Big],\\ \nonumber
 \dot{z }_{\rm rel} (\tau) &=  -\frac{1}{2\pi \ve^2 }\Big[(\Gamma_1+\Gamma_2)z _{\rm rel}\times \hat{n}_S(z _{\rm abs})  + (\Gamma_2-\Gamma_1)z_{\rm rel}\times (z_{\rm rel}\cdot \nabla\hat{n}_S(z_{\rm abs}))\\ \label{zreleqn}
 &\qquad\qquad\qquad + \mathsf{E}_{{z }_{\rm rel} }(z _{\rm abs},z _{\rm rel})\Big],\\ \label{weqn}
 \dot{w}(\tau) &=  -\frac{1}{2\pi \ve^2 } \frac{{z }_{\rm rel} }{w}\cdot   \mathsf{E}_{{z }_{\rm rel} }(z _{\rm abs},z _{\rm rel}).
 \end{align}
 
We obtain our main result, expressed in this new time:
 \begin{prop} \label{mainprop}
The curve $z_{\rm abs}(\tau)$ is an approximate magnetic geodesic:
    \begin{align}\label{zabsgeodeqn}
 \ddot{z }_{\rm abs}(\tau)&=  - \frac{  {\rm Hess} f_{{z }_{\rm abs}}(\dot{z }_{\rm abs},\dot{z }_{\rm abs})}{\|\nabla f({z }_{\rm abs})\|}  \hat{n}_S( {z }_{\rm abs}) + \frac{1}{\ve^2}(\Gamma_1 + \Gamma_2) \dot{z }_{\rm abs} \times     \hat{n}_S( {z }_{\rm abs})+ \mathsf{Error}(\tau)
  \end{align}
  where the $\mathsf{Error}$ terms are bounded according to
  \be\label{error}
|\mathsf{Error}| \lesssim  \frac{1}{\ve^4}\max\Big(\Gamma^2w^4, \ \Gamma(\Gamma_1 + \Gamma_2) w^3, \  (\Gamma_1+\Gamma_2)^2w^{2}\Big)
\ee
where  $\Gamma=  \max\{|\Gamma_1|,|\Gamma_2|\} $.
\end{prop}
\begin{remark}
All errors \eqref{error} are of commensurate order in the normalization in Remark \ref{scaling}.
\end{remark}
 
 \begin{proof}
According to \eqref{zreleqn}, we have
\begin{align*}
 \ddot{z }_{\rm abs}(\tau)  &= \frac{1}{2\pi \ve^2 }  \frac{\rmd}{\rmd \tau} \Big[(\Gamma_1-\Gamma_2)z _{\rm rel}\times \hat{n}_S(z _{\rm abs})   + \mathsf{E}_{{z }_{\rm abs} }(z _{\rm abs},z _{\rm rel})\Big]\\
 &= \frac{1}{2\pi \ve^2 }   \Big[(\Gamma_1-\Gamma_2)\dot{z} _{\rm rel}\times \hat{n}_S(z _{\rm abs})+ (\Gamma_1-\Gamma_2){z} _{\rm rel}\times\dot{z}_{\rm abs}\cdot \nabla \hat{n}_S(z _{\rm abs})  \\
 & \qquad \qquad  + \dot{z}_{\rm abs} \cdot \nabla_1 \mathsf{E}_{{z }_{\rm abs} }(z _{\rm abs},z _{\rm rel})+ \dot{z}_{\rm rel} \cdot \nabla_2 \mathsf{E}_{{z }_{\rm abs} }(z _{\rm abs},z _{\rm rel})\Big].
\end{align*}
By Lemma \ref{bndlem}, together with equations \eqref{zabseqn} and \eqref{zreleqn}, we have
\begin{align*}
|\dot{z}_{\rm abs} \cdot \nabla_1 \mathsf{E}_{{z }_{\rm abs} }(z _{\rm abs},z _{\rm rel})| &\lesssim  \frac{  \Gamma^2 }{\ve^2}w^4+  \frac{ \Gamma (\Gamma_1 + \Gamma_2)}{\ve^2} w^3,\\
|\dot{z}_{\rm rel} \cdot \nabla_2\mathsf{E}_{{z }_{\rm abs} }(z _{\rm abs},z _{\rm rel})|& \lesssim    \frac{ \Gamma (\Gamma_1 + \Gamma_2)}{\ve^2} w^3+   \frac{ (\Gamma_1 + \Gamma_2)^2}{\ve^2} w^2.
\end{align*}
Thus, using the equations \eqref{zabseqn} and \eqref{zreleqn}, we obtain
\begin{align}\nonumber
 \ddot{z }_{\rm abs}(\tau) 
&=  \frac{1}{2\pi\ve^2}(\Gamma_1 - \Gamma_2)  \dot{z }_{\rm rel}(\tau)\times \hat{n}_S( {z }_{\rm abs}) \\  \label{interzeqn}
&\qquad + \frac{1}{2\pi\ve^4}(\Gamma_1 - \Gamma_2)^2     {z} _{\rm rel} \times  ( {z} _{\rm rel} \times \hat{n}_S( {z }_{\rm abs})) \cdot \nabla \hat{n}_S( {z }_{\rm abs})+ \mathsf{Error},
  \end{align}
  with $\mathsf{Error}$ satisfying the bound \eqref{error}. To show this is (magnetic) geodesic motion, we require:
\begin{lemma}\label{gradnormlem}
Let $x\in S$ and $v$ be a unit vector orthogonal to $\hat{n}(x)$. The following identity holds
\begin{align}\label{nident}
v \times ( v\times \hat{n}_S( x)) \cdot \nabla \hat{n}_S(x)&=  -    \frac{  {\rm Hess} f_{x}(v\times \hat{n}_S, v\times \hat{n}_S)}{\|\nabla f(x)\|}\hat{n}_S(x) .
\end{align}
\end{lemma}
\begin{proof}
The gradient of the (extended) normal vector field can be expressed as
\begin{align}\nonumber
\nabla \hat{n}_S(x) &= \frac{{\rm Hess} f_x}{\|\nabla f(x)\|}- \frac{\nabla f(x)\otimes \nabla \|\nabla f(x)\|}{\|\nabla f(x)\|^2}= \frac{{\rm Hess} f_x}{\|\nabla f(x)\|}- \frac{ \hat{n}_S(x)\otimes {\rm Hess} f_x( \hat{n}_S(x), \cdot)}{\|\nabla f(x)\|} \\  \label{gradnorm}
&=\Big(I- \hat{n}_S(x)\otimes \hat{n}_S(x)\Big)\frac{{\rm Hess} f_x }{\|\nabla f(x)\|} = \mathbf{P}_x\frac{{\rm Hess} f_x }{\|\nabla f(x)\|}.
\end{align}
Next,   we express the vector $  {\rm Hess} f_{x}(v\times \hat{n}_S, \cdot)$ in the orthonormal basis  $(v, \hat{n}_S, v\times \hat{n}_S)$: 
  \begin{align}\nonumber
  {\rm Hess} f_{x}(v \times \hat{n}_S, \cdot) &=   {\rm Hess} f_{x}(v \times \hat{n}_S,v)v+   {\rm Hess} f_{x}(v \times \hat{n}_S, \hat{n}_S) \hat{n}_S\\
  &\qquad \qquad +  {\rm Hess} f_{x}(v \times \hat{n}_S,v \times \hat{n}_S) v \times \hat{n}_S.
  \end{align}
    We finally note that
$
  {\rm Hess} f_{x}(v \times \hat{n}_S,  \hat{n}_S)=0,
$
  which follows from contracting \eqref{gradnorm}  with $v \times \hat{n}_S$ and $\hat{n}_S$, and using that $\nabla \|\hat{n}_S\|^2=0$.   The identity \eqref{nident} follows.
\end{proof}
We require one more elementary lemma:
\begin{lemma}\label{normaltangentlemma}
	Let $z_1$ and $z_2\in\mathbb{R}^3$ be points on the surface $S$. Then 
	\be
	|z_{\rm rel}\cdot  \hat{n}_S( {z }_{\rm abs}) |= O(\|z_{\rm rel}\|^3).
	\ee
\end{lemma}
\begin{proof}
	By Taylor expansion
	\begin{align*}
		0 &= f(z_1) = f(z_{\rm abs} + z_{\rm rel}) = f(z_{\rm abs}) + z_{\rm rel}\cdot\nabla f(z_{\rm abs}) + (z_{\rm rel} \otimes z_{\rm rel}):\nabla^2 f(z_{\rm abs}) + O(\|z_{\rm rel}\|^3),\\
			0 &= f(z_2) = f(z_{\rm abs} - z_{\rm rel}) = f(z_{\rm abs}) - z_{\rm rel}\cdot\nabla f(z_{\rm abs}) + (z_{\rm rel} \otimes z_{\rm rel}):\nabla^2 f(z_{\rm abs}) + O(\|z_{\rm rel}\|^3).
	\end{align*}
	Statement then follows by subtracting the equations above.
\end{proof}
\noindent
Finally, we relate the velocity  $\dot{z }_{\rm abs}$ to $\hat{z}_{\rm rel} \times \hat{n}_S( {z }_{\rm abs})$ and  $\dot{z }_{\rm rel}$. From \eqref{zabseqn} and \eqref{zreleqn}, we have
\begin{align*}
 \frac{1}{2\pi\ve^2} (\Gamma_1-\Gamma_2)z _{\rm rel}\times \hat{n}_S(z _{\rm abs}) &= \dot{z }_{\rm abs}(\tau)  -  \frac{1}{2\pi\ve^2}\mathsf{E}_{{z }_{\rm abs} }(z _{\rm abs},z _{\rm rel})\\
 &=
  \dot{z }_{\rm abs} +  O\left( \Gamma\frac{w^3}{\ve^2}\right) + O\left( (\Gamma_1+\Gamma_2)\frac{ w^2}{\ve^2}\right)
  \end{align*}
by \eqref{Eabs1}. We have also
 \begin{align*}
 (\Gamma_1 - \Gamma_2)\dot{z }_{\rm rel}(\tau)&=-(\Gamma_1+\Gamma_2)\dot{z }_{\rm abs}(\tau)  \\
 &\qquad +\frac{1}{2\pi\ve^2}(\Gamma_1 - \Gamma_2)^2 z_{\rm rel}\times(z_{\rm rel} \cdot \nabla \hat{n}_S (z_{\rm abs}))\\
 &\qquad\qquad +\frac{1}{2\pi\ve^2}\Big[   (\Gamma_1 + \Gamma_2)\mathsf{E}_{{z }_{\rm abs} }(z _{\rm abs},z _{\rm rel}) -  (\Gamma_1 - \Gamma_2)\mathsf{E}_{{z }_{\rm rel} }(z _{\rm abs},z _{\rm rel})\Big].
 \end{align*}
 Using the bounds \eqref{Erelest} and \eqref{Eabs1}, we find
 \begin{align}\nonumber
 	\label{rel-abs}
 	(\Gamma_1 - \Gamma_2)\dot{z }_{\rm rel} &=- (\Gamma_1+\Gamma_2)  \dot{z }_{\rm abs}
 	- \frac{1}{2\pi\ve^2}\left( (\Gamma_1 - \Gamma_2)^2 z_{\rm rel}\times(z_{\rm rel} \cdot \nabla \hat{n}_S (z_{\rm abs}))\right)\\
 	& \quad+ O\left(\Gamma^2\frac{ w^4}{\ve^2}\right) +   O\left(\Gamma (\Gamma_1+\Gamma_2) \frac{w^2}{\ve^2}\right).
 \end{align}
 Finally, plugging \eqref{rel-abs} and \eqref{nident} into \eqref{interzeqn}, we obtain
  \begin{align}\nonumber
 	\ddot{z }_{\rm abs}(\tau) 
 	&=  \frac{1}{2\pi\ve^2}(\Gamma_1 - \Gamma_2)  \dot{z }_{\rm rel}(\tau)\times \hat{n}_S( {z }_{\rm abs}) \\  \nonumber
 	&\qquad -\frac{1}{(2\pi)^2\ve^4}(\Gamma_1 - \Gamma_2)^2      \frac{  {\rm Hess} f_{x}({z }_{\rm rel}\times \hat{n}_S, {z }_{\rm rel}\times \hat{n}_S)}{\|\nabla f(x)\|}\hat{n}_S(x) + \mathsf{Error}\\
 	\nonumber
 	&=  -\frac{1}{2\pi\ve^2}(\Gamma_1 + \Gamma_2)  \dot{z }_{\rm abs}(\tau)\times \hat{n}_S( {z }_{\rm abs})  \\
 	\nonumber
 	&\qquad	- \frac{1}{(2\pi)^2\ve^4}\left( (\Gamma_1 - \Gamma_2)^2 z_{\rm rel}\times(z_{\rm rel} \cdot \nabla \hat{n}_S (z_{\rm abs}))\right)\times\hat{n}_S( {z }_{\rm abs})\\
 	&\qquad\qquad-   \frac{  {\rm Hess} f_{x}(\dot{z}_{\rm abs}, \dot{z}_{\rm abs})}{\|\nabla f(x)\|}\hat{n}_S(x) + \mathsf{Error}.
 \end{align}
 Expanding the triple product in the second term, we observe that it can be absorbed into $\mathsf{Error}$:
 \begin{align} \nonumber
 	&\frac{1}{(2\pi)^2\ve^4}\left( (\Gamma_1 - \Gamma_2)^2 z_{\rm rel}\times(z_{\rm rel} \cdot \nabla \hat{n}_S (z_{\rm abs}))\right)\times\hat{n}_S( {z }_{\rm abs})  \\ \nonumber
 	&\qquad =\frac{1}{(2\pi)^2\ve^4}(\Gamma_1 - \Gamma_2)^2 \left( (\hat{n}_S( {z }_{\rm abs})\cdot z_{\rm rel})z_{\rm rel}\cdot\nabla \hat{n}_S (z_{\rm abs}) - (\hat{n}_S (z_{\rm abs}) \cdot(z_{\rm rel} \cdot \nabla \hat{n}_S (z_{\rm abs}))z_{\rm rel})\right)\\
 	& \qquad  = O\left(\Gamma^2 \frac{w^4}{\ve^4}  \right),
 \end{align}
 where we used Lemma \ref{normaltangentlemma} and the fact that $\hat{n}_S \cdot(z_{\rm rel} \cdot \nabla \hat{n}_S) = \frac{1}{2}z_{\rm rel}\cdot\nabla\|\hat{n}_S\|^2 = 0$.
Hence we get
 \begin{align}\nonumber
 \ddot{z }_{\rm abs}(\tau) 
&=  \frac{1}{2\pi\ve^2}(\Gamma_1 - \Gamma_2)  \dot{z }_{\rm rel}(\tau)\times \hat{n}_S( {z }_{\rm abs}) \\  \nonumber
&\qquad -\frac{1}{(2\pi)^2\ve^4}(\Gamma_1 - \Gamma_2)^2      \frac{  {\rm Hess} f_{x}({z }_{\rm rel}\times \hat{n}_S, {z }_{\rm rel}\times \hat{n}_S)}{\|\nabla f(x)\|}\hat{n}_S(x) + \mathsf{Error}\\
&=  -\frac{1}{2\pi\ve^2}(\Gamma_1 + \Gamma_2)  \dot{z }_{\rm abs}(\tau)\times \hat{n}_S( {z }_{\rm abs})  -   \frac{  {\rm Hess} f_{x}(\dot{z}_{\rm abs}, \dot{z}_{\rm abs})}{\|\nabla f(x)\|}\hat{n}_S(x) + \mathsf{Error}.
  \end{align}
 The term involving $ {\rm Hess} f$ in \eqref{zabsgeodeqn} can be expressed in terms of the second fundamental form of the surface via Lemma \ref{secondfundformlem}. This establishes Proposition \ref{mainprop}.
 \end{proof}
   In view of \eqref{weqn}, in the limit $\ve\to 0$, we have $\dot{w}(\tau) = O\left(\frac{1}{\ve^2}\Gamma w^4\right) + O\left(\frac{1}{\ve^2}(\Gamma_1 + \Gamma_2) w^3\right)$ so
  \be
  \lim_{\ve \to 0} \frac{w(\tau)}{\ve}=1, \qquad \forall \tau\in \mathbb{R},
  \ee
  for any scaling (e.g. that of \eqref{scalingrequirements}) such that $\mathsf{Error}(0) \lesssim  \max\Big(\Gamma^2, \ \frac{\Gamma(\Gamma_1 + \Gamma_2)}{\ve} , \  \frac{(\Gamma_1+\Gamma_2)^2}{\ve^2}\Big)\to 0$.
Thus  ${z }_{\rm abs}(\tau)\to X(\tau)$ and $\tau(t)\to t$ as $\ve\to 0$. This completes the proof of Theorem \ref{mainthm}.
  
\newpage

\appendix
\section{Estimates for error terms}\label{appendix}
Here we prove  necessary estimates for validity of Lemma \ref{lem:hard}.

\begin{lemma}\label{bndlem}
 With $\Gamma=  \max\{|\Gamma_1|,|\Gamma_2|\} $ and $ \mathsf{E}_{{z }_{\rm rel} }, \mathsf{E}_{{z }_{\rm abs} }$ as in \eqref{Edefs}, the following bounds hold
	\begin{align}
	\label{Erelest}
	  \mathsf{E}_{{z }_{\rm rel} }(a,b) &=  O(\Gamma \|b\|^4)+O((\Gamma_1+\Gamma_2) \|b\|^3), \\  
	   \label{Eabs1}
		  \mathsf{E}_{{z }_{\rm abs} }(a,b) &= O(\Gamma \|b\|^3)+ O( (\Gamma_1+\Gamma_2)\|b\|^2), \\ \label{Eabs2}
		\nabla_a \mathsf{E}_{{z }_{\rm abs} }(a,b) &= O(\Gamma \|b\|^3)   + O((\Gamma_1+\Gamma_2) \|b\|^2),\\ \label{Eabs3}
		\nabla_b\mathsf{E}_{{z }_{\rm abs} }(a,b) &= O(\Gamma \|b\|^2)  + O((\Gamma_1+\Gamma_2) \|b\|),
	\end{align}
	where implicit constants depend on the $ \|\mathsf{Reg}_S(\cdot,\cdot)\|_{C^{2,\alpha}}$ and the curvature of the surface $\|K(\cdot)\|_{L^\infty(S)}$.
\end{lemma}
 \begin{proof}
 First we note that the terms appearing in \eqref{f1} and \eqref{f2} are explicitly
  \begin{align*}
\frac{1}{\Gamma_1} \nabla_p \mathsf{Reg}_S(p,q)&=- \frac{\Gamma_2}{2\pi}\nabla_1 \log \frac{d_S(p,q)}{\|p-q\|}
+  \Gamma_2  \nabla_1 ( R_S(p,q)- R_S(p,p))
 + (\Gamma_1+\Gamma_2) \nabla_1 R_S(p,p), \\
\frac{1}{\Gamma_2} \nabla_q \mathsf{Reg}_S(p,q)&=- \frac{\Gamma_1}{2\pi}\nabla_2 \log \frac{d_S(p,q)}{\|p-q\|}+  \Gamma_1  \nabla_2  (R_S(p,q)- R_S(q,q))
 +(\Gamma_1+ \Gamma_2) \nabla_2 R_S(q,q).
\end{align*}
since
$\nabla_p R_S(p,p)  = \nabla_1 R_S(p,p)  + \nabla_2 R_S(p,p) = 2\nabla_1 R_S(p,p)= 2\nabla_2 R_S(p,p)$.
Moreover, interpreting $p$ and $q$ as points in the ambient $\mathbb{R}^3$, by Taylor's theorem  (see \cite[Appendix A]{nicolaescu2012random}):
\be
d_S^2(p,q) = \|p-q\|^2 +  \frac{ K(p)}{6} |(p- q)\cdot (p+q)^\perp|^2 \|p-q\|^2 + O( \|p- q\|^5),
\ee
so that
\be
 \log \frac{d_S(p,q)}{\|p-q\|}= \frac{1}{2} \log \frac{d_S^2(p,q)}{\|p-q\|^2} =  \log \Big(1 +  \frac{ K(p)}{6} |(p- q)\cdot (p+q)^\perp|^2 + O( \|p- q\|^3)\Big).
\ee
Below, we will use that that for a function $f(z_1,z_2)$ we have the identities:
   \begin{align*}
(\nabla_1 f)(a+b, a-b) &= \frac{1}{2}\Big(\nabla_a f(a+b, a-b) +   \nabla_b f(a+b, a-b) \Big),\\
(\nabla_2 f)(a+b, a-b) &= \frac{1}{2}\Big(\nabla_a f(a+b, a-b) -   \nabla_b f(a+b, a-b) \Big).
 \end{align*}
We begin analyzing $   \mathsf{E}_{{z }_{\rm rel} }(a,b)$. We have
    \begin{align*}
    \mathsf{E}_{{z }_{\rm rel} }(a,b) &= F_1(a,b)+  F_2(a,b) - (\Gamma_2-\Gamma_1)b\times (b\cdot \nabla\hat{n}_S(a))\\
    &= b\times\Big(\Gamma_2 \hat{n}_S(a+b) -(\Gamma_1+ \Gamma_2) \hat{n}_S(a) +\Gamma_1 \hat{n}_S(a-b) -(\Gamma_2-\Gamma_1)b\cdot \nabla\hat{n}_S(a)\Big)\\
    &\qquad     +\tfrac{\pi}{\Gamma_2}  \|b\|^2 (\nabla_q \mathsf{Reg}_S)(a+b, a-b)\times \hat{n}_S(a-b) -    \tfrac{\pi}{\Gamma_1}  \|b\|^2 (\nabla_p \mathsf{Reg}_S)(a+b, a-b)\times \hat{n}_S( a+b),\\
        &= b\times\Big(\Gamma_2 \hat{n}_S(a+b) -(\Gamma_1+ \Gamma_2) \hat{n}_S(a) +\Gamma_1 \hat{n}_S(a-b) -(\Gamma_2-\Gamma_1)b\cdot \nabla\hat{n}_S(a)\Big)\\
    &\qquad     -\frac{1}{2} \|b\|^2  \left[\Gamma_1\nabla_2 \log \left(\frac{d_S^2(a+b,a-b)}{\|2b\|^2}\right) -\Gamma_2\nabla_1 \log \left(\frac{d_S^2(a+b,a-b)}{\|2b\|^2}\right)  \right]\\
    &\qquad \quad +  \pi \|b\|^2  \Big[ \Gamma_1   (\nabla_2 R_S(a+b,a-b)- \nabla_2 R_S(a-b,a-b))\\
    & \qquad \qquad \qquad \qquad \qquad -  \Gamma_2   (\nabla_1 R_S(a+b,a-b)- \nabla_1 R_S(a+b,a+b)) \Big] \\
    &\qquad \qquad +4 \pi (\Gamma_1+ \Gamma_2) \|b\|^2 \Big[ \nabla_2 R_S(a-b,a-b)- \nabla_1 R_S(a+b,a+b)\Big]\\
    &= {\tt I} + {\tt II} + {\tt III} +  {\tt IV}.
        \end{align*}
        We proceed to estimate term by term.
        We first have (bounds depend on $\|\hat{n}\|_{C^2(S)}$)
        \begin{align*}
      {\tt I} &= b\times\Big(\Gamma_2 \hat{n}_S(a+b) -(\Gamma_1+ \Gamma_2) \hat{n}_S(a) +\Gamma_1 \hat{n}_S(a-b)-(\Gamma_2-\Gamma_1)b\cdot \nabla\hat{n}_S(a)\Big)\\
       &= b\times\Big(\Gamma_2(\hat{n}_S(a+b) -\hat{n}_S(a)-b\cdot \nabla\hat{n}_S(a)) + \Gamma_1(\hat{n}_S(a-b) -\hat{n}_S(a)+b\cdot \nabla\hat{n}_S(a))\Big)\\
              &=O( (\Gamma_1 + \Gamma_2) \|b\|^3) + O(\Gamma\|b\|^4).
        \end{align*}
 Next, for term $ {\tt II}$, we manipulate
        \begin{align} \nonumber
        \Gamma_1\nabla_2 \log &\left(\frac{d_S^2(a+b,a-b)}{\|2b\|^2}\right) -\Gamma_2\nabla_1 \log \left(\frac{d_S^2(a+b,a-b)}{\|2b\|^2}\right) \\ \nonumber
        &=  \Gamma_1 \nabla_a \log \left(\frac{d_S^2(a+b,a-b)}{\|2b\|^2}\right) - (\Gamma_1+ \Gamma_2)\nabla_1 \log \left(\frac{d_S^2(a+b,a-b)}{\|2b\|^2}\right)     \\
        &= O(\Gamma \|b\|^2)+ O((\Gamma_1+ \Gamma_2)\|b\|) .\label{logbnd}
        \end{align}
        Thus we find that
        \be
        {\tt II} = O((\Gamma_1+\Gamma_2)\|b\|^3)+ O(\Gamma \|b\|^4).
        \ee
        Finally, Taylor expanding we have
        \begin{align*}
         \nabla_1 R_S(a+b,a-b) &=    \nabla_1R_S(a,a)+ b\cdot \nabla_1\nabla_1 R_S(a,a) - b\cdot \nabla_2\nabla_1 R_S(a,a)   + O(\|b\|^2),\\
        \nabla_2 R_S(a+b,a-b) &=    \nabla_2 R_S(a,a)+ b\cdot \nabla_1\nabla_2 R_S(a,a) - b\cdot \nabla_2\nabla_2 R_S(a,a) + O(\|b\|^2),\\
              \nabla_1 R_S(a+b,a+b) &=   \nabla_1 R_S(a,a)+ b\cdot \nabla_1\nabla_1 R_S(a,a) +b\cdot \nabla_2\nabla_1 R_S(a,a)  + O(\|b\|^2),\\
              \nabla_2 R_S(a-b,a-b) &=   \nabla_2 R_S(a,a)- b\cdot \nabla_1\nabla_2 R_S(a,a) - b\cdot \nabla_2\nabla_2 R_S(a,a) + O(\|b\|^2),
        \end{align*}
        Thus we have that the differences satisfy
      \begin{align}\label{d2iden1}
         \nabla_2 R_S(a+b,a-b)- \nabla_2 R_S(a-b,a-b) &= 2 b\cdot \nabla_1\nabla_2 R_S(a,a)+ O(\|b\|^2),\\ \label{d1iden1}
         \nabla_1 R_S(a+b,a-b)- \nabla_1 R_S(a+b,a+b) &= -2 b\cdot \nabla_2\nabla_1 R_S(a,a) + O(\|b\|^2),\\ \label{d2iden1}
         \nabla_2 R_S(a-b,a-b)- \nabla_1 R_S(a+b,a+b) &= O(\|b\|),
       \end{align}
       where we used symmetry $ \nabla_1 R_S(a,a)= \nabla_2 R_S(a,a)$.
Thus, we have obtained:
      \begin{align}
       {\tt III} &=   2\pi \|b\|^2 (\Gamma_1+\Gamma_2) b\cdot \nabla_2\nabla_1 R_S(a,a)+ O(\Gamma\|b\|^4) = O((\Gamma_1+\Gamma_2)\|b\|^3)+ O(\Gamma \|b\|^4), \\
        {\tt IV} &=  O((\Gamma_1+\Gamma_2)\|b\|^3).
      \end{align}

Next we estimate $ \mathsf{E}_{{z }_{\rm abs} }(a,b)$.  We write
        \begin{align*}
        \mathsf{E}_{{z }_{\rm abs} }(a,b) &=  F_2(a,b)-F_1(a,b)\\
        &= \Gamma_1b\times\Big( \hat{n}_S(a-b)-2\hat{n}_S(a)+\hat{n}_S(a+b) \Big) -  (\Gamma_1+\Gamma_2)b\times\Big( \hat{n}_S(a+b) - \hat{n}_S(a) \Big)  \\
        &\quad +   \tfrac{\pi}{\Gamma_2}  \|b\|^2 (\nabla_q \mathsf{Reg}_S)(a+b, a-b)\times \hat{n}_S(a-b)+ \tfrac{\pi}{\Gamma_1}  \|b\|^2 (\nabla_p \mathsf{Reg}_S)(a+b, a-b)\times \hat{n}_S( a+b)\\
      &= \Gamma_1b\times\Big( \hat{n}_S(a-b)-2\hat{n}_S(a)+\hat{n}_S(a+b) \Big) -  (\Gamma_1+\Gamma_2)b\times\Big( \hat{n}_S(a+b) - \hat{n}_S(a) \Big)  \\
    &\qquad     -\frac{1}{2} \|b\|^2  \left[\Gamma_1\nabla_2 \log \left(\frac{d_S^2(a+b,a-b)}{\|b\|^2}\right) +\Gamma_2\nabla_1 \log \left(\frac{d_S^2(a+b,a-b)}{\|b\|^2}\right)  \right]\\
    &\qquad \quad +  \pi \|b\|^2  \Big[ \Gamma_1   (\nabla_2 R_S(a+b,a-b)- \nabla_2 R_S(a-b,a-b))\\
    & \qquad \qquad \qquad \qquad \qquad + \Gamma_2   (\nabla_1 R_S(a+b,a-b)- \nabla_1 R_S(a+b,a+b)) \Big] \\
    &\qquad \qquad +4 \pi (\Gamma_1+ \Gamma_2) \|b\|^2 \Big[ \nabla_2 R_S(a-b,a-b)+ \nabla_1 R_S(a+b,a+b)\Big]\\
    &= {\tt V} + {\tt VI} + {\tt VII} +  {\tt VIII}.
    \end{align*}
  We have
  \begin{align}\label{nbd2d}
   |\Gamma_1b\times\Big( \hat{n}_S(a-b)-2\hat{n}_S(a)+\hat{n}_S(a+b) \Big) |&\leq \Gamma \|\hat{n}\|_{C^2(S)} \|b\|^3,\\
   | (\Gamma_1+\Gamma_2)b\times\Big( \hat{n}_S(a+b) - \hat{n}_S(a) \Big)| &\leq (\Gamma_1+\Gamma_2)  \|\hat{n}\|_{C^1(S)} \|b\|^2.
  \end{align}
  Thus we find 
  \be
  {\tt V} =  O((\Gamma_1+\Gamma_2)\|b\|^2)+ O(\Gamma \|b\|^3).
  \ee
The other terms are easily estimated by Taylor expansion to satisfy
\begin{align}
    {\tt VI},\   {\tt VII},   \     {\tt VIII}  &=  O((\Gamma_1+\Gamma_2)\|b\|^2)+ O(\Gamma \|b\|^3).
\end{align}
  The stated bound \eqref{Eabs1} follows.  The derivative estimates \eqref{Eabs2} and \eqref{Eabs3} follow similarly.
  \end{proof}

 \subsection*{Acknowledgments}  We are grateful to M. Khuri for fruitful discussions. TDD and DG were partially supported by the NSF DMS-2106233 grant and NSF CAREER award \#2235395.  BK  was partially supported by an NSERC Discovery Grant.
 
 \vspace{-2mm}
 \subsection*{Data Availability Statement}  No data was created or analyzed in this study.

\bibliographystyle{amsplain}
\bibliography{magnetic_dipole}

\end{document}